\documentclass[reqno,12pt]{amsart} 
\usepackage
{amsaddr}
\usepackage{setspace,graphicx,epstopdf,amsmath,amsfonts,amsgen,
amstext,amsthm,amsbsy,amsopn,amssymb,txfonts,bbm,epic, pstricks,tikz,pict2e,parskip,verbatim,mathrsfs}  
\usepackage[square,numbers]{natbib} 
\usepackage[utf8]{inputenc}   
\usepackage[pdfborder={0 0 0}]{hyperref}

\usepackage{dsfont}

\let\oldsqrt\sqrt
\def\sqrt{\mathpalette\DHLhksqrt}
\def\DHLhksqrt#1#2{%
\setbox0=\hbox{$#1\oldsqrt{#2\,}$}\dimen0=\ht0
\advance\dimen0-0.2\ht0
\setbox2=\hbox{\vrule height\ht0 depth -\dimen0}%
{\box0\lower0.4pt\box2}}
 
\def\b0{\boldsymbol{0}}

\def\bcdot{\boldsymbol{\cdot}}

 \newcommand{\C}     {\mathbb{C}} 
\newcommand{\R}     {\mathbb{R}} 
\newcommand{\Z}     {\mathbb{Z}} 
\newcommand{\N}     {\mathbb{N}}

\newcommand{\E}     {\mathbb{E}} 
 
\newcommand{\T}     {\mathbb{T}}

\newcommand{\Acal}   {{\mathcal A }}
\newcommand{\Bcal}   {{\mathcal B }}

\newcommand{\Ecal}   {{\mathcal E }} 
\newcommand{\Fcal}   {{\mathcal F }} 
\newcommand{\Gcal}   {{\mathcal G }}

\newcommand{\Mcal}   {{\mathcal M }}

\def\Afrak{\mathfrak{A}}
\def\qfrak{\mathfrak{q}}

 \newcommand{\com}{\textnormal{c}}

\newcommand{\tr}{{\operatorname {Tr}}}
\newcommand{\Exp}{\mathscr{E}\kern-0.2mm{\operatorname{xp}}}
\newcommand{\Log}{\mathscr{L}\kern-0.2mm{\operatorname{og}}}

\def\ex{{\rm e}}

\newcommand{\abs}[1]{{\lvert #1\rvert}}

\def\1{{\mathchoice {1\mskip-4mu\mathrm l}      
{1\mskip-4mu\mathrm l} 
{1\mskip-4.5mu\mathrm l} {1\mskip-5mu\mathrm l}}}

\theoremstyle{definition}

\theoremstyle{remark}

\theoremstyle{plain}
\numberwithin{equation}{section}

\newtheorem{classmain}{Theorem}[section]

\newtheorem{quantmain}[classmain]{Theorem}
\newtheorem{subadd}[classmain]{Lemma}
\newtheorem{incompat}[classmain]{Lemma}
\newtheorem{quantumestimates1}[classmain]{Lemma}
\newtheorem{quantumestimates2}[classmain]{Lemma}
\newtheorem{quantumestimates3}[classmain]{Lemma}
\newtheorem{quantumestimates4}[classmain]{Lemma}
\newtheorem{quantumRP}[classmain]{Lemma}

\definecolor{forestgreen}   {cmyk}{0.91, 0   , 0.88, 0.12}
 
\newcommand{\rk}[1]{{\color{forestgreen}{#1}}}

\title[Staggered long-range order for diluted quantum spin models]{Staggered long-range order\\ for diluted quantum spin models}
\author{Roman Koteck\'y\textsuperscript{1,2} and Benjamin Lees\textsuperscript{3}}

\begin{document}
\maketitle
\begin{center}
\textsuperscript{1} Mathematics Institute, University of Warwick, Coventry CV4 7AL, United Kingdom

\textsuperscript{2} CTS, Charles University, Prague, Czech Republic

\textsuperscript{3} Fachbereich Mathematik, Technische Universit\"at Darmstadt, Germany

\email{r.kotecky@warwick.ac.uk},  \ \   
\email{lees@mathematik.tu-darmstadt.de}
\end{center}

\maketitle

\date{22 August, 2018}

\begin{abstract}
\noindent
We study an annealed site diluted
quantum XY model with  spin $S\in\tfrac12\N$.
We find regions of the parameter space where, in spite of  being a priori favourable for a densely occupied state,  phases with staggered occupancy occur at low temperatures. 

\end{abstract}

\section{Introduction}
The quantum XY model with  spin $S\in\tfrac12\N$ on the square lattice $\Z^2 $
with a particular type of annealed site dilution is considered. 
We prefer to formulate the  model  in terms of a more symmetric equivalent version, with dilution represented by Ising spins instead of the site occupation numbers, with the~Hamiltonian 
\begin{equation}
\label{E:Hq}
H
=-\frac{1}{S^2}\sum_{\{ x,y\}}\sigma_x\sigma_y(S_x^1S_y^1+S_x^3S_y^3-S(S+1))-\kappa\sum_{\{ x,y\}}\sigma_x\sigma_y-\mu\sum_{x}\sigma_xS_x^3.
\end{equation}
Here  $S_{x}^j$'s are the standard spin-$S$ operators acting on the site $x$ (so in particular $S^1$ and $S^3$ are real matrices and $S^3$ is a diagonal matrix) and $\sigma_{x}$ is an Ising variable representing the particle at the site $x$ (the occupancy  number $n_x\in\{0,1\}$ indicating the presence/absence of a particle at $x$ corresponds to $\sigma_x=2 n_x-1 \in\{-1,1\}$).
The parameters   $\mu$  and $\kappa$  
allude to the chemical potential and the interaction parameter for the particles. 

Our main claim concerns the  existence  of a staggered long range occupancy order  
characterised by the presence of two distinct states (in the thermodynamic limit) which preferentially take Ising spin with value $+1$ on either the even or the odd sublattice. Indeed, it will be proven that such  states  occur in a region of parameters $\mu$ and $\kappa$, at intermediate inverse temperatures, $\beta$. 

The existence of such states can be viewed as a demonstration of an ``effective entropic repulsion'' caused by the interaction of  quantum spins leading to an impactful restriction of the ``available phase space volume''. As a result, occupation of adjacent sites might turn out to be unfavourable---it results  in an effective repulsion between particles in nearest neighbour sites and as a result eventually leads  to a staggered order.  
It is easy to understand  that this is  the case for the  annealed site diluted  Potts model with large number of spin states $q$ \cite{C-K-S}
where this effect is indeed caused by a pure entropic repulsion: two nearest neighbour occupied sites contribute the Boltzmann factor 
$q+q(q-1)e^{-\beta}$ which is at low temperatures much smaller than the factor $q^2$ obtained from two next nearest neighbour spins
that are free to take  all possible spin values entirely independently.  
Actually, the same is true---even though less obvious---in the case of diluted models with classical continuous spins  \cite{C-K-S2}.
Our result constitutes an extension  of  similar claims to a quantum situation.

To get a control on effective repulsion, we rely on a standard tool---the chessboard estimates which follow from reflection positivity.  The  classical references on this topic are \cite{D-L-S,F-I-L-S,F-I-L-S2,F-L,  F-S-S} with a recent review  \cite{B}. For our case the treatment in \cite{B-C-S} is especially useful. In particular, we use the setting from \cite[Section 3.3]{B-C-S} for an efficient formulation of the long range order in terms of  coexistence of the corresponding infinite-volume KMS states.

Note that we could also add a term $-uS_x^2S_y^2$ to our Hamiltonian and our result concerning reflection positivity would still hold for $u\leq 0$ (as $S^2$ is a purely imaginary matrix), we can consider our case as restricting ourselves only to the case $u=0$, the results might extend to models with $-1\leq u<0$ however it is not clear what estimates we can obtain in these cases (see Lemmas \ref{emptystaggeredestimate}-\ref{keyestimate}).

We introduce the model and state the main result in  Section \ref{sec settingandmainresults}. The  proof is  deferred to  Section \ref{sec quantproof}.
 
\section{Setting and Main Results}
\label{sec settingandmainresults}
For a fixed \emph{even} $L\in \N$, we  consider the \emph{torus} $\T_L=\Z^d/L\Z^d$ consisting of $L^d$ sites that can be identified with the set 
$(-L/2,L/2]^d\cap\Z^d$. On the torus $\T_L$ we take the \emph{algebra $\Afrak_L$ of observables} consisting of all functions $A: \{-1,1\}^{\T_L}\to \Mcal_L$ where $\Mcal_L$ is s the $C^*$-algebra of linear operators acting on the space 
$\otimes_{x\in\T_L}\C^{2S+1}$ with $S\in\frac{1}{2}\N$ (complex  $(2S+1)^{\abs{\T_L}}$-dimensional matrices). 

A particular example of an observable is the Hamiltonian $H_L\in \Afrak_L$ of the form \eqref{E:Hq} with the periodic boundary conditions (on the torus $\T_L$),
\begin{multline}
\label{E:HL}
H_L(\sigma)=-\frac{1}{S^2}\sum_{\{ x,y\}}\sigma_x\sigma_y(S_x^1S_y^1+S_x^3S_y^3-S(S+1))-\kappa\sum_{\{ x,y\}}\sigma_x\sigma_y-\mu\sum_{x\in\T_L}\sigma_xS_x^3.
\end{multline}
Here the sum is over pairs $\{x,y\}\in\E_L$, the set of all edges connecting nearest neighbour sites in the torus $\T_L$, and $(S^1,
S^3)$ are (two of) the spin-$S$ matrices.
The Gibbs state on the torus is given by
\begin{equation}\label{Gibbsstate}
\langle\cdot\rangle_{L,\,\beta}=\frac{1}{Z_L(\beta)}\sum_{\sigma}\tr\,\cdot \ex^{-\beta H_L}
\end{equation}
with $Z_L(\beta)=\sum_{\sigma}\tr\, \ex^{-\beta H_L}$.
Infinite volume states  of a quantum spin system are formulated in terms of KMS states,
an analog of DLR states for classical systems.
Let us briefly recall this notion in the form to be used in our situation. Here we follow closely the treatment from \cite{B-C-S} which can be consulted  for a more detailed discussion of KMS states in a setting similar to ours. Let $\Afrak$ denote the \emph{$C^*$ algebra of quasilocal observables},
\begin{equation}
\Afrak=\overline{\Afrak_0},\qquad\text{ where }\quad \Afrak_0=\bigcup_{\Lambda\subset\Z^d \text{ finite}}\Afrak_\Lambda,
\end{equation}
where the overline denotes the norm-closure. We define the \emph{time evolution operators} $\alpha_t^{(L)}$ acting on $A\in\Afrak_L$ and for any $t\in \R$ as
\begin{equation}
\alpha_t^{(L)}(A)=\ex^{it H_L}A\ex^{-itH_L}.
\end{equation}
It is well known that for a local operator $A\in\Acal_0$ we can expand $\alpha^{(L)}_t(A)$ as a series of commutators,
\begin{equation}
\alpha^{(L)}_t(A)=\sum_{m\geq 0}\frac{(it)^m}{m!}[H_L,[H_L,...,[H_L,A]...]].
\end{equation}
The map $t\to  \alpha^{(L)}_t$ extends to all $t\in\C$ and, as $L\to\infty$, $\alpha_t^{(L)}$ converges in norm to an operator $\alpha_t$ on $\Afrak$ uniformly on compact subsets of $\C$ (one can consult the proof, for example,  in \cite{R} and see that the same proof structure works in this case). A state $\langle\cdot\rangle_\beta$ on $\Afrak$ (a positive linear functional ($\langle  A\rangle_\beta\ge 0$ if $A\ge 0$) such that $\langle   \mathds{1}\rangle_\beta= 1$)
is called a KMS state (or is said to satisfy the KMS condition) with a Hamiltonian $H$  at an inverse temperature $\beta$, if we have
\begin{equation}\label{KMScond}
 \langle AB\rangle_\beta=\langle\alpha_{-i\beta}(B)A\rangle_\beta
\end{equation}
for the above defined family of operators $\alpha_t$ at imaginary values $t=-i\beta$.
One can see that the Gibbs state \eqref{Gibbsstate} satisfies the KMS condition for the finite volume time evolution operator.

A special class of \emph{observables} are classical events  $\mathds{1}_{\Fcal} I$ obtained as a product 
of the identity $I\in \Mcal_L$ with the indicator $\mathds{1}_{\Fcal}$ of an Ising configuration event  $\Fcal\subset \{-1,1\}^{\T_L}$.
Often we will consider (classical) block  events  depending only on the Ising configuration on the block-cube of $2^d$ sites, $C=\{0,1\}^d\subset \T_L$. Namely, the events of the form $\Ecal\times \{0,1\}^{\T_L\setminus C}$
where $\Ecal \subset \{0,1\}^{C}$. We will refer to these events  directly as block events $\Ecal$ and use a streamlined 
 notation $\langle\Ecal\rangle_{L,\,\beta}$ (resp. $\langle\Ecal\rangle_{\beta}$)  instead of 
$\langle\mathds{1}_{\Ecal\times \{0,1\}^{\T_L\setminus C}} I\rangle_{L,\,\beta}$ (resp. $\langle\mathds{1}_{\Ecal\times \{0,1\}^{\T_L\setminus C}} I\rangle_{\beta}$).

In particular, to characterise the long-range order states mentioned above, we introduce the block events 
$\Gcal^{\mspace{1mu}\textnormal{e}}= \{\sigma^{\textnormal{e}}\}$ and $ \Gcal^{\mspace{1mu}\textnormal{o}}= \{\sigma^{\textnormal{o}}\}$
where $\sigma^{\textnormal{e}}$ and $\sigma^{\textnormal{o}}$ are the even and the odd staggered configurations on $C$:
$\sigma^{\textnormal{e}}_{x}=1$ iff $x$ is an even site in $C$ and $\sigma^{\textnormal{e}}_{x}=1$ iff $x$ is an odd site in $C$. 
 Notice that the sets $\Gcal^{\mspace{1mu}\textnormal{e}}$ and $\Gcal^{\mspace{1mu}\textnormal{o}}$ are disjoint.

The main result for the quantum system with Hamiltonian \eqref{E:HL} can now be stated as follows.

\begin{quantmain}\label{quantumchessboardstates}
Let $d=2$ and $S\geq \frac{1}{2}$. Let $\mu_0=\frac12 \frac{S+1}{S^2}$ and $\kappa_0=\kappa_0(\mu)=\frac{S+1}{S}-2\abs{\mu} S$. Then, for any  $|\mu|<\mu_0$, $\kappa<\kappa_0(\mu)$, and any $0<\varepsilon<\frac{1}{2}$,  there exists $\beta_0=\beta_0(\mu,\kappa,\varepsilon)$ such that 
 for 
any $\beta>\beta_0$ there exist two distinct  KMS states, $\langle\cdot\rangle_\beta^{\textnormal{e}}$ and $\langle\cdot\rangle_\beta^{\textnormal{o}}$, that are staggered,
\begin{equation}
\langle \Gcal^{\mspace{1mu}\textnormal{e}}\rangle_\beta^{{\textnormal{e}}}\geq 1-\varepsilon	
\text{ and  }
\langle \Gcal^{\mspace{1mu}\textnormal{o}}\rangle_\beta^{{\textnormal{o}}}\geq 1-\varepsilon.
\end{equation}
\end{quantmain}
The proof of this theorem is the content of Section \ref{sec quantproof}. For the technical estimates, we are restricting ourselves to  the two-dimensional case  $d=2$. The proof of a similar claim for $d> 2$ (with other $\mu_0$ and $\kappa_0$ depending on $d$) employing the same methods is straightforward but rather cumbersome. 

Notice that for $|\mu|<\mu_0$ we have $\kappa_0(\mu)>0$. It is not so surprising that 
 that the claim is true for any negative $\kappa$---negative $\kappa$ should trigger antiferromagnetic staggered
order at low temperatures. More interesting is the case, established  by the theorem,  when this happens for positive $\kappa$ where it is a demonstration of an effective  entropic repulsion stemming from the quantum spin.  

\section{Proof of Theorem \ref{quantumchessboardstates}}
\label{sec quantproof}
\subsection{Reflection Positivity for the Annealed Quantum Model}

Consider now a splitting of the torus $\T_L$ into two disjoint halves, $\T_L=\T_L^+\cup\T_L^-$,  separated by a pair of planes; for example say, $P_1=\{(-1/2,x_2,\dots,x_d)$ and $P_2=
\{(L/2-1/2,x_2,\dots,x_d)$, $x_2,\dots, x_d\in \R$.
We introduce a reflection $\theta :\T_L\to \T_L$ defined by $\theta x=(-(x_1+1),x_2,\dots,x_d) $.
\footnote{Notice that on the torus, 
the reflection with respect to $P_1$ is identical with that with respect to $P_2$ (just notice that  $\abs{x_1-(-1/2)}= \abs{y_1-(-1/2)} $ with $x_1\neq y_1$ implies $y_1=-(x_1+1)$,
while $\abs{x_1-(L/2-1/2)}= \abs{y_1-(L/2-1/2)}$ with $x_1\neq y_1$ implies $y_1=-(x_1+L+1)$ and $-(x_1+1)= -(x_1+L+1) \mod(L)$.}
Any such reflection (parallel $P_1$ and $P_2$ of distance $L/2$ in arbitrary  half-integer position and orthogonal to any coordinate axis) will be called  \emph{reflections through planes between the sites} or simply \emph{reflections} (we will not use the other reflections through planes on the sites that are useful for classical models). Notice that $\theta$ maps $\T_L^+$ into $ \T_L^- $ and $\theta^2=1$.

Further,  consider an algebra  $\Afrak_L$  with two subalgebras  $\Afrak_L^+, \Afrak_L^-\subset \Afrak_L$,
$\Afrak_L=\Afrak_L^+\otimes \Afrak_L^-$,   living on the sets $\T_L^+, \T_L^- $, respectively.
Namely, we define $\Afrak_L^+$ as a set of all operator-valued functions $A: \{-1,1\}^{\T_L^+}\to \Mcal_L^+$, where
$\Mcal_L^+$ is the set of all   operators of the form $I\otimes A^+$ with $A^+$
acting on the subspace $\otimes_{x\in\T_L^+}\C^{2S+1}$ and $ I$ is the identity on the complementary space $\otimes_{x\in\T_L^-}\C^{2S+1}$.  Similarly for $ \Afrak_L^-$.

The reflection $\theta :\T_L^-\to \T_L^+$ can be naturally elevated to a morphism  $\theta: \Afrak_L^+\to \Afrak_L^-$ (cf.  twisted reflections in \cite[Section 3.4]{F-I-L-S}) with
$\theta$ flipping the spin in the Ising configuration and rotating by $\pi$ in the second coordinate direction of spins $S_x$. More precisely, define the unitary operator 
\begin{equation}\label{eq:Uop}
U=\prod_{x\in \T_L^{\rk{-}}} e^{i\pi S_x^2}
\end{equation}
on the subspace $\otimes_{x\in\T_L^-}\C^{2S+1}$
and, for $\sigma\in\{-1,1\}^{\T_L}$, define $\theta \sigma$ by 
\begin{equation}
\label{E:sigma}
(\theta \sigma)_x=-\sigma_{\theta x}.
\end{equation}
 Then for $A\in\Afrak_L^+$ 
with $A(\sigma) = I \otimes A^+(\sigma)$ 
for any  $\sigma\in \T_L^+$,  we define the operator $\theta A\in \Afrak_L^-$ by  
\begin{equation}
\label{E:thetaA}
\theta A(\sigma)=\overline{U^{-1}A^+(\theta \sigma)U}\otimes I, \sigma \in \T_L^-  .
\end{equation}
Here $\overline A$ denotes the complex conjugation of the operator $A$.

Note the effect of the reflection on  spin operators:  for any $j\in\{1,2,3\}$ and $x\in\T_L^+$,  we have 
$\overline{U^{-1}S_x^jU}=-S_x^j$ and thus
\footnote{Actually, the Hamiltonian \eqref{E:HL} depends only on the spin operators  $S_x^1$ and $S_x^3$.
 Their standard representation is by real matrices and thus the the complex conjugation in  \eqref{E:thetaA} can be skipped for them.}
 \begin{equation}
 \label{E:S}
\theta S_x^j= - S_{\theta x}^j.
\end{equation}
Similarly, for the operator $A(\sigma)=S^3_x\sigma_x$, we have 
 \begin{equation}
  \label{E:Ssigma}
\theta A(\sigma)=(-S^3_{\theta x})(-\sigma_{\theta x})=S^3_{\theta x}\sigma_{\theta x}
\end{equation}
and for the operator $A(\sigma)= \sigma_x i I$ with $i I$ the multiple of a unit matrix by the imaginary unit $i$, we have 
 \begin{equation}
  \label{E:iI}
\theta A(\sigma)=(-\sigma_{\theta x})(-i I )= i \sigma_{\theta x}\, I.
\end{equation}

Finally, we say that a state $\langle \bcdot\rangle$ on $\Afrak_L $ is reflection positive with respect to  $\theta$ 
if for any $A, B\in \Afrak_L^+$ we have
\begin{equation}
\label{E:RPineq}
\left\langle A \theta B\right\rangle= \left\langle B\theta A\right\rangle
\end{equation}
and 
\begin{equation}
\label{E:RPineq2}
\left\langle A\theta A\right\rangle\geq 0.
\end{equation}
The standard consequence of the reflection positivity is the Cauchy-Schwarz inequality
\begin{equation}\label{E:CS}
\langle A \theta B\rangle^2\leq\langle A \theta A\rangle\langle B \theta B \rangle
\end{equation}
for any $A, B\in \Afrak_L^+$.

In our situation of an annealed diluted quantum model, we are dealing with the state
\begin{equation}
  \left\langle A\right\rangle_{L,\,\beta}= 
\frac{ \sum_{\sigma\in \{-1,1\}^{\T_L}}  \tr\, A(\sigma)\ex^{-\beta H_L(\sigma)}}{\sum_{\sigma\in \{-1,1\}^{\T_L}}\tr\, \ex^{-\beta H_L(\sigma)}}
 \end{equation}
for any $A\in \Afrak_L$ and with the Hamiltonian $H_L\in \Afrak_L$ of the form \eqref{E:HL}. 

The standard proof of reflection positivity may be extended to this case.

\vspace{0.3cm}
\begin{quantumRP} 
The state  $ \langle \bcdot \rangle_{L,\,\beta}$   is reflection positive for any $\theta $ through planes between the sites
and any  $\mu\in\mathbb{R}$, $\kappa\leq \tfrac{S+1}{S}$ and $\beta\ge 0$.   
\end{quantumRP}
\begin{proof}
The equality \eqref{E:RPineq} is immediate. For \eqref{E:RPineq2} 
 we first write the Hamiltonian $H_L$ in the form $H_L(\sigma,\theta  \sigma')= J(\sigma)+ \theta J(\sigma') -\sum_\alpha D_\alpha(\sigma)\, \theta  D_\alpha(\sigma')$ for any $\sigma,\sigma'\in \{-1,1\}^{\T_L^+}$
where $J\in \Afrak_L^+$ consists of  all terms of the Hamiltonian with (both) sites in $\T_L^+$ and $D_\alpha  \theta  D_\alpha$, with 
$D_\alpha\in \Afrak_L^+$   indexed by $\alpha$, represent the terms corresponding to edges crosses the reflection plane. 

Indeed, we define
\begin{equation}
J(\sigma)=-\frac{1}{S^2}\sum_{\substack{\{ x,y\}\\x,y\in\T_L^+}}\sigma_x\sigma_y(S_x^1S_y^1+S_x^3S_y^3-S(S+1))-\kappa\sum_{\substack{\{ x,y\}\\x,y\in\T_L^+}}\sigma_x\sigma_y-\mu\sum_{x\in\T_L^+}\sigma_xS_x^3
\end{equation}
and note that, due to the definition of $\theta$, $\theta J(\sigma)$ is the same as $J(\sigma)$ but with $\T_L^+$ replaced by $\T_L^-$. This is clear for the first two sums as we pick up four resp. two factors of $-1$, for the last term note that we also pick up two factors of $-1$, one from $\theta S^1_x=-S^1_{\theta x}$ and one from $\theta \sigma_x=-\sigma_{\theta x}$.
 If $\{x,y\}$ is an edge crossing the reflection plane (i.e. $x\in\T_L^x$, $y=\theta x\in\T_L^-$), the corresponding $D_\alpha$'s are
\begin{align}
D^0_x=&\sqrt{\tfrac{S+1}{S}-\kappa}\, i\, \sigma_x
\\
D^1_x=&\frac1S \sigma_x\,S_x^1
\\
D^3_x=&\frac1S \sigma_x\,S_x^3
\end{align}
If $\kappa\leq \tfrac{S+1}{S}$, we have 
\begin{equation}
(\tfrac{S+1}{S}-\kappa)\, \sigma_x\sigma_{y}= -D_x^0\,  \theta (D_x^0)
\end{equation}
since,
in view of \eqref{E:sigma} and \eqref{E:iI},
\begin{equation}
\sigma_x\sigma_{y}=- i \sigma_x\, i \sigma_y= - i \sigma_x\, \theta(i \sigma_x).
\end{equation}
Also $\sigma_xS^{j}_x\sigma_y S^{j}_y=\sigma_xS^{j}_x  \theta(\sigma_xS^{j}_x)$ for $j=1,3$.

 For the claim \eqref{E:RPineq2}  we need to show that 
 \begin{equation}
 \sum_{\sigma, \sigma'\in \{-1,1\}^{\T_L^+}}\tr\, A(\sigma) \theta A(\sigma')\ex^{-\beta H_L(\sigma,\theta \sigma')}\geq 0
 \end{equation}
 for any $A\in \Afrak_L^+$.
Adapting  the standard proof, see e.g. \cite[Theorem 2.1]{F-L},  
by Trotter's formula we get
 \begin{equation}
  \ex^{-\beta H_L(\sigma,\theta \sigma')}=\lim_{k\to\infty}\Bigl(\ex^{-\frac{\beta}{k}J(\sigma)}\ex^{-\frac{\beta}{k} \theta J(\sigma')}\bigl[1+\tfrac{\beta}{k}\sum\nolimits_\alpha D_\alpha (\sigma) \theta D_\alpha(\sigma')\bigr]\Bigr)^k=:\lim_{k\to\infty}F_k(\sigma,
 \sigma').
 \end{equation}
The needed claim will be verified once show that
\begin{equation}\label{E:RPproof}
 \sum_{\sigma,\sigma'\in \{-1,1\}^{\T_L^+}}\tr\, \left(A(\sigma\theta A(\sigma')\,F_k(\sigma,\sigma') \right)\geq 0 
\end{equation}
for all $k$. 

Indeed, proceeding exactly in the same way as in the proof of Theorem 2.1 in \cite{F-L}, we can conclude that for each $\sigma,\sigma' \in \{-1,1\}^{\T_L^+}$ the operator  $F_k(\sigma,\sigma')$ can be written as a sum of terms of the form  $F_k^{(\ell)}(\sigma) \theta F_k^{(\ell)}(\sigma')$, where $F_k^{(\ell)}\in\Afrak_L^+$. Each such term yields
 \begin{multline}
 \sum_{\sigma,\sigma'\in \{-1,1\}^{\T_L^+}}\tr(A(\sigma)  \theta A(\sigma') F_k^{(\ell)}(\sigma) \theta F_k^{(\ell)}(\sigma')=\\=
  \sum_{\sigma,\sigma'\in \{-1,1\}^{\T_L^+}}   \tr(A(\sigma) F_k^{(\ell)}(\sigma)\theta (AF_k^{(\ell)})(\sigma')=  \Biggl|\sum_{\sigma\in \{-1,1\}^{\T_L^+}}\tr\bigl(A(\sigma) F_k^{(\ell)}(\sigma)\bigr)\Biggr|^2\ge 0
\end{multline}
thus completing the proof.
\end{proof}

\subsection{Chessboard estimates}

Consider $\T_L$ partitioned into $(L/2)^d$ disjoint $2\times 2\times\dots\times 2$ blocks $C_{\tau}\subset \T_L$ labeled by vectors $\tau \in \T_{L/2}$ with $2\tau$ denoting the position of their lower left corner. Clearly, $C_{\tau}=C+2 \tau$
with $C_{\b0}=C$. 

If $\tau\in\T_{L/2}$ with $|\tau|=1$, we let $\theta _{\tau}$ be the reflection with respect to the plane between $C$ and  $C_{\tau}$ corresponding to $\tau$. Further, 
if $\Ecal$ is a block event, $\Ecal\subset\{-1,1\}^C$,  we let $\vartheta_{\tau}(\Ecal)\subset\{-1,1\}^{C_{\tau}}$ be the correspondingly reflected event,
$\sigma\in \Ecal$ iff $\theta\sigma\in\vartheta_{\tau}(\Ecal)$. For other $\tau$'s in $\T_{L/2}$ we define $\vartheta_{\tau}(\Ecal)$ by a sequence of reflections (note that the result does not depend on the choice of sequence leading from $C$ to $C_{\tau}$.). If all coordinates  of $\tau$ are even this simply results in the translation by $2\tau$. 

Chessboard estimates are   formulated in terms of a mean value of a homogenised pattern based on a block event $\Ecal$
disseminated throughout the lattice,
\begin{equation}
\mathfrak{q}_{L,\,\beta}(\Ecal):=\Bigl(\Bigl\langle\prod_{\tau\in\T_{L/2}}\vartheta_{\tau}(\Ecal)\Bigr\rangle_{L,\,\beta}\Bigr)^{(2/L)^d}.
\end{equation}
If $\kappa\leq \tfrac{S+1}{S}$,   $\Ecal_1,...,\Ecal_m$ are block events, and $\tau_1,...,\tau_m\in\T_{L/2}$ are distinct,  we get,
by a standard repeated use of reflection positivity,  the  chessboard estimates
\begin{equation}\label{chessboardestimate}
 \Bigl\langle\prod_{j=1}^m\vartheta_{\tau}(\Ecal_j)\Bigr\rangle_{L,\,\beta}\leq \prod_{j=1}^m\Bigl(\Bigl\langle\prod_{\tau\in\T_{L/2}}\vartheta_{\tau}(\Ecal_j)\Bigr\rangle_{L,\,\beta}\Bigr)^{(2/L)^d}=\prod_{j=1}^m\mathfrak{q}_{L,\,\beta}(\Ecal_j).
\end{equation}
Note that we have chosen to split $\T_L$ into $2\times 2\times\dots\times2$ blocks with the bottom left corner of the basic block $C$ at the origin $(0,0,\dots,0)$. If we had instead replaced  the basic block $C$  by its shift $C+e_1$ by the unit vector $e_1=(1,0,\dots,0)$,
the same estimate would hold with  the new partition with all blocks shifted by $e_1$. We will use this fact in the sequel.

The proof of the  useful property of  subadditivity of the function  $\mathfrak{q}_{L,\,\beta}$ for classical systems \cite[Lemma 5.9]{B} can be also directly extended  to our case.
\begin{subadd}\label{subadd}
 Suppose $\kappa\leq \tfrac{S+1}{S}$. If $\Ecal, \Ecal_1,\Ecal_2,...$ are events  on $C $ such that $\Ecal\subset\cup_k \Ecal_k$, then 
 \begin{equation}
  \mathfrak{q}_{L,\,\beta}( \Ecal)\leq \sum_k \mathfrak{q}_{L,\,\beta}(\Ecal_k).
 \end{equation}
\end{subadd}
\vspace{-0.3cm}
\begin{proof} Using subadditivity of $\langle \cdot\rangle_{L,\,\beta}$, we get
\begin{equation}
  \mathfrak{q}_{L,\,\beta}(\Ecal)^{(L/2)^d}
  =\bigg\langle\prod_{\tau\in\T_{L/2}}\vartheta_{\tau}(\Ecal)\bigg\rangle_{L,\,\beta}\le
\sum_{(k_{\tau})}\bigg\langle\prod_{\tau\in\T_{L/2}}\vartheta_{\tau}(\Ecal_{k_{\tau}})\bigg\rangle_{L,\,\beta}
\end{equation}
Using now the chessboard estimate
\begin{equation}
\bigg\langle\prod_{\tau\in\T_{L/2}}\vartheta_{\tau}(\Ecal_{k_{\tau}})\bigg\rangle_{L,\,\beta}\le
\prod_{\tau\in\T_{L/2}}   \mathfrak{q}_{L,\,\beta}(\Ecal_{k_{\tau}}),
\end{equation}
we get
\begin{multline}
  \mathfrak{q}_{L,\,\beta}(\Ecal)^{(L/2)^d}\le
\sum_{(k_{\tau})}\prod_{\tau\in\T_{L/2}} \mathfrak{q}_{L,\,\beta}(\Ecal_{k_{\tau}})=\\=
\prod_{\tau\in\T_{L/2}}\biggl(\sum_{k} \mathfrak{q}_{L,\,\beta}(\Ecal_{k})\biggr)
=\biggl(\sum_{k} \mathfrak{q}_{L,\,\beta}(\Ecal_{k})\biggr)^{(L/2)^d}.
\end{multline}
\end{proof}
Let us introduce the set $\Bcal$ of bad configurations, $\Bcal=\{-1,1\}^C\setminus( \Gcal^{\mspace{1mu}\textnormal{e}}\cup \Gcal^{\mspace{1mu}\textnormal{o}})$,
and use  $\tau_{r}$ to denote the shift by $r\in\T_L$.
The proof of the existence of two distinct  KMS states is based on the following lemma. 

\begin{incompat}\label{imcompatcondition}
There exists functions $\mu_0,\kappa_0$  as  stated in Theorem \ref{quantumchessboardstates} such that for any $\varepsilon>0$,  $\mu$ such that $|\mu|<\mu_0$ and $\kappa<\kappa_0(\mu)$  
there exists $\beta_0$ such that 
 for  any $\beta>\beta_0$, any $L$ sufficiently large, and any distinct $\tau_1,\tau_2\in\T_{L}$,
\begin{align}
\langle& \Bcal\rangle_{L,\,\beta}<\varepsilon, \label{badblocks}
\\
\langle& \tau_{2\tau_1}(\Gcal^{\mspace{1mu}\textnormal{e}})\cap\tau_{2\tau_2}(\Gcal^{\mspace{1mu}\textnormal{o}})\rangle_{L,\,\beta}<\varepsilon. \label{incompatgoodblocks}
\end{align}
\end{incompat}
Deferring its proof to the next section, we show here how it implies  Theorem \ref{quantumchessboardstates}.\\
\emph{Proof of Theorem \ref{quantumchessboardstates} given Lemma \ref{imcompatcondition}}.
We closely follow the proof of Lemma 4.5 and Proposition 3.9 in \cite{B-C-S}. Define 
\begin{equation}
\T_L^{\textnormal{front}}=\{x\in\T_L : -\lfloor L/4-1/2\rfloor \leq x_1\leq \lceil L/4-1/2\rceil\}.
\end{equation}
We denote by $\mathfrak{A}_L^{\textnormal{front}}$ the algebra of observables localised in $\T_L^{\textnormal{front}}$.

Let $\Delta_M\subset\T_{L/2}$ be a $M\times M$ block of sites on the ``back'' of $\T_{L/2}$ (dist$(0,\Delta_M)\geq L/4 -M)$. Then for a block event $\Ecal$ depending only on the Ising configuration in $C$ define
\begin{equation}
 \rho_{L,M}(\Ecal)=\frac{1}{|\Delta_M|}\sum_{\tau\in\Delta_M}\tau_{2\tau}(\Ecal).
\end{equation}
If $\langle\Ecal\rangle_{L,\,\beta}\geq c$ for all $L\gg 1$ for a constant $c>0$ then we can define a new state on $\mathfrak{A}^{\textnormal{front}}_L$, by
\begin{equation}
\langle \cdot \rangle_{L,M;\beta}=\frac{\langle \rho_{L,M}(\Ecal)\;\cdot\;\rangle_{L,\,\beta}}{ \langle\rho_{L,M}(\Ecal)\rangle_{L,\,\beta}}.
\end{equation}
We claim that if $\langle \,\,\rangle_\beta$ is a weak limit of $\langle\,\, \rangle_{L,M;\beta}$ as $L\to\infty$ and then $M\to\infty$ then $\langle \,\,\rangle_\beta$ is a KMS state at inverse temperature $\beta$ invariant under translations by $2\tau$ for $\tau\in\T_{L}$.

Indeed translation invariance comes from the spatial averaging in $\rho_{L,M}(\Ecal)$. As in \cite{B-C-S} we need to show that $\langle\,\, \rangle_{\beta}$ satisfies the KMS condition \eqref{KMScond}. For an observable  $A$ on the `front' of the torus, $\T_L^{\textnormal{front}}$, we have
\begin{equation}
[\alpha_t^{(L)}(A),\rho_{L,M}(\Ecal)]\to 0 \text{ as } L\to\infty
\end{equation}
in norm topology uniformly for $t$ in compact subsets of $\C$. Using this and \eqref{KMScond} for the finite volume Gibbs states we have that for $A,B$ bounded operators on the ``front'' of the torus
\begin{equation}
\langle \rho_{L,M}(\Ecal) AB\rangle_{L,\,\beta}=\langle \rho_{L,M}(\Ecal) \alpha^{(L)}_{-i\beta}(A)B\rangle_{L,\,\beta} +o(1) \text{ as } L\to\infty.
\end{equation}
Because $\alpha^{(L)}_{-i\beta}(B)\to\alpha_{-i\beta}(B)$ as $L\to\infty$ in norm we have that $\langle\,\,\rangle_{L,M;\beta}$ converges as $L\to\infty$ and then $M\to\infty$ to a KMS state at inverse temperature $\beta$.

The proof of Theorem \ref{quantumchessboardstates} follows by taking $\Ecal=\Gcal^{\mspace{1mu}\textnormal{e}}$ or $\Ecal=\Gcal^{\mspace{1mu}\textnormal{o}}$ as we know both staggered configurations have the same expectation we can define a state $\langle \,\, \rangle^{\textnormal{e}}_{L,M;\beta}$, using Lemma \ref{imcompatcondition} we conclude that $\langle\rho_{L,M}(\Gcal^{\mspace{1mu}\textnormal{e}})\rangle_{L,\,\beta}$ is uniformly positive and hence
\begin{equation}
\langle \tau_{2\tau}(\Gcal^{\mspace{1mu}\textnormal{e}}) \rangle^{\textnormal{e}}_{L,M;\beta}\geq 1-\varepsilon,
\end{equation}
for any $\tau\in\T_L^{\textnormal{front}}$ (if $M\ll L/2$) and similarly for $\langle \,\, \rangle^{\textnormal{o}}_{L,M;\beta}$. If $\varepsilon$ is small enough then the right-hand side of this inequality will be greater than $1/2$, hence in the thermodynamic limit $\Gcal^{\mspace{1mu}\textnormal{e}}$ will dominate. \qed

To prove Lemma \ref{imcompatcondition} we use a version of Peierls' argument hinging on chessboard estimates.

\subsection{Peierls' argument}

For a given Ising configuration, consider the event $\tau_{2\tau_1}(\Gcal^{\mspace{1mu}\textnormal{e}})\cap \tau_{2\tau_2}(\Gcal^{\mspace{1mu}\textnormal{o}})$ that the blocks $C_{\tau_1}$ and $C_{\tau_2}$ have different staggered configurations described by $\Gcal^{\mspace{1mu}\textnormal{e}}$ and $\Gcal^{\mspace{1mu}\textnormal{o}}$ respectively. 
The idea is to show the existence of a contour separating the points $\tau_1$ and $\tau_2$ and to use chessboard estimates to show that occurrence of such a contour is improbable. 

Consider the set of all  blocks (labeled by) $\tau\in\T_{L/2}$ such that a translation of the even staggered configuration $\tau_{2\tau}(\Gcal^{\mspace{1mu}\textnormal{e}})$ occurs on it.
Let $\Delta\subset\T_{L/2}$ be its connected component containing $\tau_1$. Consider the component $\overline\Delta\subset\T_{L/2}$
of $\Delta^\com$ containing $\tau_2$. The set of edges $\gamma$ of the graph $\T_{L/2}$ between vertices of $\overline\Delta$ 
and its complement $\overline\Delta^\com$ is a minimal cutset of $\Delta$. Informally, $\gamma$ is a contour between   $\Delta$ with all its holes except the one containing $\tau_2$ filled up and  the remaining component containing $\tau_2$--- \emph{a contour separating $\tau_1$ and $\tau_2$}.
The standard fact is that the number of contours with a fixed number of edges $\abs{\gamma}=n$ separating two vertices $\tau_1$ and $\tau_2$ is bounded by $c^n$ with a suitable constant $c$.

Given a contour $\gamma$ of length $\abs{\gamma}=n$, there exists a coordinate direction such that there are at least $n/d$ edges in $\gamma$ aligned along this direction. 
Precisely half of them have their outer endpoint (the vertex in $\overline\Delta$) ``on the left'' of its inner endpoint, choosing (arbitrarily) the direction of the chosen coordinate axis (without loss of generality we can take for this the first coordinate axis) as $e_1$,
there are at least $n/(2d)$ edges $\{\tau, \tau+e_1\}$ such that $\tau\in \overline\Delta$ and $\tau+e_1\in \Delta$.

Now, the crucial claim is that with each contour we can associate at least  $1/2$ of the $n/(2d)$ bad blocks (with a configuration from $\vartheta_{2\tau}(\Bcal)$), all belonging to  a given fixed partition: either to our original partition of $\T_L$ labelled by $\T_{L/2}$ or    to a new partition of $\T_L$ with the basic block  $C$ shifted by a unit vector from $\T_L$ in direction $e_1$.
 Indeed, any block corresponding to an outer vertex  $\tau$  above is either bad or, if  not, it has to be a translation  $\tau_{2\tau}(\Gcal^{\mspace{1mu}\textnormal{o}})$ of the odd staggered configuration (being the even staggered configuration would be in contradiction with the assumption that $\Delta$ is a connected component of the set of blocks with even staggered configuration). However, then the  block shifted by a unit vector in $\T_L$ in direction $e_1$ features an odd staggered configuration on its left-hand half and an even staggered configurations on its right-hand half, i.e., a configuration that belongs to the properly shifted set $\Bcal$
(here it is helpful that the set $\Bcal$ is invariant with respect to the reflection through the middle plane of the block).

We use $S(\gamma)$ to denote this collection of at least $\abs{\gamma}/(4d)$ bad blocks  associated with contour $\gamma$.
Given that, according to the construction above, all blocks from $S(\gamma)$ belong to the same partition (either the original one or a shifted one), we can use the chessboard estimate based on the the corresponding partition to bound
the probability that all blocks of a given set $S(\gamma)$
are bad  by
\begin{equation}
\label{E:gammaextchessboard}
\bigg\langle \prod_{\tau\in S(\gamma)}\vartheta_{\tau}(\Bcal)\bigg\rangle_{L,\,\beta}\leq \mathfrak{q}_{L,\,\beta}(\Bcal)^{\abs{S(\gamma)}}.
\end{equation}

As a result,  assuming that  $\mathfrak{q}_{L,\,\beta}(\Bcal)\leq 1$ (we will later show it can be made arbitrarily small),  the expectation of the event $ \tau_{2\tau_1}(\Gcal^{\mspace{1mu}\textnormal{e}})\cap \tau_{2\tau_2}(\Gcal^{\mspace{1mu}\textnormal{o}})$ is bounded by 
\begin{equation}
\bigg\langle \tau_{2\tau_1}(\Gcal^{\mspace{1mu}\textnormal{e}})\cap \tau_{2\tau_2}(\Gcal^{\mspace{1mu}\textnormal{o}})\bigg\rangle_{L,\,\beta}\le \sum_{\gamma \text{ separating } \tau_1 \text{ and } \tau_2} \mathfrak{q}_{L,\,\beta}(\Bcal)^{\abs{\gamma}/(4d)} 2^{\abs{\gamma}/(2d)+1}.
\end{equation}
Here, $2^{\abs{\gamma}/(2d)+1}	$ is the bound on the number of sets $S(\gamma)$ associated with the contour $\gamma$ once the direction $e_1$ is chosen.

This leads to the final bound
\begin{equation}
\bigg\langle \tau_{2\tau_1}(\Gcal^{\mspace{1mu}\textnormal{e}})\cap \tau_{2\tau_2}(\Gcal^{\mspace{1mu}\textnormal{o}})\bigg\rangle_{L,\,\beta}\le 
\sum_{n=4}^\infty 2\bigl(4\mathfrak{q}_{L,\,\beta}(\Bcal)^{n/(4d)} \bigr)c^n.
\end{equation}

We now see that Lemma \ref{imcompatcondition} will hold if $\mathfrak{q}_{L,\,\beta}(\Bcal)$ can be made arbitrarily small by tuning the parameters of the model correctly. Hence we turn our attention to this.

For the remaining technical part of this section we restrict ourselves to the \emph{two-dimensional case}.

   For $d=2$, the set $\Bcal$ consists of 14 configurations that can be classified into  five events according to the number of sites in $C$ that have Ising spin $+1$, $\Bcal=\Bcal^{(0)}\cup\Bcal^{(1)}\cup\Bcal^{(2)}\cup\Bcal^{(3)}\cup\Bcal^{(4)}$.
Here, $\Bcal^{(0)}$ and $\Bcal^{(4)}$ consist of a single configuration (fully $-1$ and fully $+1$, respectively) and $\Bcal^{(1)},\Bcal^{(2)},\Bcal^{(3)}$
consist each of 4 configurations related by symmetries.
Notice that the event  $\Bcal^{(2)}$ has precisely two $+1$ spins at neighbouring positions (excluding the configurations  $\sigma^{\textnormal{e}}$ and $\sigma^{\textnormal{o}}$). 

%

By subadditivity we can bound $\mathfrak{q}_{L,\,\beta}(\Bcal)$ by the sum of expectations of  homogenised patterns based on the fourteen configurations from $\Bcal$  disseminated throughout the lattice by reflections. In view of the symmetries, we need only consider only 5 configurations $\sigma^{(k)}, k=0,1,\dots,4$, one from each event $\Bcal^{(k)}, k=0,1,\dots,4$. In fact we can see that, as reflections flips the sign of Ising variables, that we need only consider $k=0,1,2$ Indeed, the dissemination of pattern $\Bcal^{(0)}$  differs from the dissemination of pattern $\Bcal^{(4)}$ by a shift by $2e_1$, and the dissemination of pattern $\Bcal^{(1)}$  differs from the dissemination of pattern $\Bcal^{(3)}$ by a shift by $2e_1$ and a rotation.

We use  $Z^{(k)}_L(\beta)$ to denote the corresponding quantities
\begin{equation}
 Z^{(k)}_L(\beta)=\mathfrak{q}_{L,\,\beta}(\{\sigma^{(k)}\})^{(L/2)^2}Z_L(\beta),
\end{equation}
for $k\in\{0,1,\dots,4\}$. For notational consistency we also denote the contribution of staggered configurations on $\T_L$ as $Z^{\mspace{1mu}(\textnormal{e})}_L(\beta)$ and $ Z^{\mspace{1mu}(\textnormal{o})}_L(\beta)$ \vspace{0.3cm}
\begin{quantumestimates1}\label{emptystaggeredestimate}
For any $\mu\in\mathbb{R}$ and $\kappa<\kappa_0(\mu)$ we have 
  \begin{align}
 Z_L^{(0)}(\beta),Z_L^{(4)}(\beta)\leq& e^{\beta L^2 |\mu |\ S } \tr \exp\left\{\frac{\beta}{S^2}\sum_{\{ x,y\}}(S_x^1S_y^1+S_x^3S_y^3)\right\}, \label{eq:B0bound}
 \\
  Z_L^{(1)}(\beta),Z_L^{(2)}(\beta),Z_L^{(3)}(\beta)\leq&e^{\beta L^2 \left(|\mu|\ S-\kappa+\tfrac{S+1}{S}\right) }  \tr \exp\left\{\frac{\beta}{S^2}\sum_{\{ x,y\}}(S_x^1S_y^1+S_x^3S_y^3)\right\}, \label{eq:B12bound}
  \\
Z^{\mspace{1mu}(\textnormal{e})}_L(\beta),Z^{\mspace{1mu}(\textnormal{o})}_L(\beta)\geq & e^{\beta L^2\left(-|\mu|S-2\kappa+2\tfrac{S+1}{S}\right)}\tr \exp\left\{\frac{\beta}{S^2}\sum_{\{ x,y\}}(S_x^1S_y^1+S_x^3S_y^3)\right\} \label{eq:Bebound}
 \end{align}
\end{quantumestimates1}
\begin{proof}
We begin by removing the terms  associated to $S(S+1),\kappa$ and $\mu$ from the Hamiltonian, i.e., we need bounds on the terms $(-\tfrac{S+1}{S}+\kappa)\sum_{\{ x,y\}}\sigma_x^{(k)}\sigma_y^{(k)}$
and $\mu\sum_{x\in\T_L}\sigma^{(k)}_xS_x^3$ (occuring in $-H$), for $\sigma^{(k)}$, the Ising configuration corresponding to the disseminated pattern $\Bcal^{(k)}$. 

For the first term we use that $\sigma_x^{(k)}\sigma_y^{(k)}=\pm 1$ for each  $\{x,y\}$.
In particular, we get $\sum_{\{ x,y\}}\sigma_x^{(k)}\sigma_y^{(k)}=0$ for $k=0,4$,
it equals $-L^2$ for $k=1,2,3$, and it equals
$-2L^2$  for $k=\textnormal{e},\textnormal{o}$.
Indeed, 
for $\sigma^{(0)}$ and $\sigma^{(4)}$ half of the links yield $-1$ (they are are between a \emph{plus} and a  \emph{minus}) and the second half yield $+1$. For $\sigma^{(1)}$, $\sigma^{(2)}$, and $\sigma^{(3)}$ three quarters of the links yield $-1$ and  one quarters $+1$. Finally, for $k=\textnormal{e}$ and $k=\textnormal{o}$ all links yield $-1$.

For the $\mu$-term we use the simple bound
 \begin{equation}
 -|\mu|\ S L^2\leq \mu\left\|\sum_{x\in\T_L}\sigma_xS_x^3\right\|\leq |\mu| S L^2.
 \end{equation}
 
 Together this gives the factors in front of the traces in equations \eqref{eq:B0bound}, \eqref{eq:B12bound}, and \eqref{eq:Bebound}. What remains in each case is a term of the form 
\begin{equation}
-\frac{1}{S^2}\sum_{\{ x,y\}}\sigma^{(k)}_x\sigma^{(k)}_y(S_x^1S_y^1+S_x^3S_y^3)
\end{equation}
where $k\in\{0,1,...4,e,o\}$. By conjugating with a unitary operator acting as $e^{i\pi S^2}$ on the sites where $\sigma^{(k)}_x=-1$ we can turn this operator into,
\begin{equation}
-\frac{1}{S^2}\sum_{\{ x,y\}}(S_x^1S_y^1+S_x^3S_y^3).
\end{equation}
As we have conjugated by a unitary operator this conjugation does not affect the trace. This completes the proof.
\end{proof}

As a result, we get the following bounds on the expectations of the disseminated bad configurations $\mathfrak{q}_{L,\,\beta}(\{\sigma^{(k)}\})$ for $k=0,1,\dots,4$.

\begin{quantumestimates4}\label{keyestimate}
 Let $\mu\in\mathbb{R}$ and $\kappa<\kappa_0(\mu)$. We have
 \begin{align}
 \qfrak_{L,\beta}(\{\sigma^{(0)}\}),\qfrak_{L,\beta}(\{\sigma^{(4)}\})\leq &2^{-4/L^2}\exp\left\{4\beta\left(2|\mu| S+2\kappa-2\tfrac{S+1}{S}\right)\right\}
\\
\qfrak_{L,\beta}(\{\sigma^{(1)}\}),\qfrak_{L,\beta}(\{\sigma^{(2)}\}),\qfrak_{L,\beta}(\{\sigma^{(3)}\})\leq &2^{-4/L^2}\exp\left\{4\beta\left(2|\mu| S+\kappa-\tfrac{S+1}{S}\right)\right\}
\end{align}
 \end{quantumestimates4}  
 
\begin{proof}
All the estimates follow from the previous lemmas using
\begin{equation}
\mathfrak{q}_{L,\,\beta}(\{\sigma^{(k)}\})=\left(\frac{Z_L^{(k)}(\beta)}{Z_L(\beta)}\right)^{(2/L)^2}\leq\left(\frac{Z_L^{(k)}(\beta)}{2 Z^{\mspace{1mu}\textnormal{e}}_L(\beta)}\right)^{(2/L)^2}.
\end{equation}
\end{proof}

Further, using subadditivity (Lemma \ref{subadd}) we have
\begin{equation}\label{subaddbad}
\mathfrak{q}_{L,\,\beta}(\Bcal)\leq \mathfrak{q}_{L,\,\beta}(\{\sigma^{(0)}\}) +4 \sum_{k=1}^3 \mathfrak{q}_{L,\,\beta}(\{\sigma^{(k)}\})
+\mathfrak{q}_{L,\,\beta}(\{\sigma^{(4)}\}).
\end{equation}
From Lemma \ref{keyestimate} we can see that for $\beta$ large this quantity will be small if 
\begin{equation}
\kappa < \min\{1+\tfrac{1}{S}-|\mu| S,1+\tfrac{1}{S}-2|\mu|S\}=1+\tfrac{1}{S}-2|\mu| S=:\kappa_0(\mu).
\end{equation}
This condition is compatible with the requirement $\kappa\leq 1+\tfrac{1}{S}$ in Lemma \ref{subadd} and allows us to take $\kappa>0$ once $|\mu|<\tfrac{1}{2S}+\tfrac{1}{2S^2}$.

More precisely, we see that there exists $\mu_0>0$ and a  function $\kappa_0$   that is positive on $(-\mu_0,\mu_0)$ such that if 
$|\mu|<\mu_0$, $\kappa<\max(\kappa_0(\mu),0)$, and  $\varepsilon>0$,  there exists $\beta_0(\mu,\kappa,\varepsilon)$ such that 
the claims of Lemma \ref{imcompatcondition} and  thus also Theorem \ref{quantumchessboardstates} are valid for any $\beta \ge \beta_0$.

\section*{Acknowledgement}
The research of R.K. was  supported by the grant GA\v CR 16-15238S and that of B.L.  partially by EPSRC grant EP/HO23364/1 and partially by the Alexander von Humboldt Foundation. R.K. would also like to thank Isaac Newton Institute for Mathematical Sciences for hospitality during the programme \emph{Scaling limits, rough paths, quantum field theory} (supported by EPRSC Grant Number EP/R014604/1) where the work on the final version of the paper was undertaken.

\end{document}